\PassOptionsToPackage{table}{xcolor}
\documentclass[runningheads]{llncs}
\usepackage{graphicx}
\usepackage{marvosym}
\usepackage{tikz}
\usetikzlibrary{positioning}
\usetikzlibrary{datavisualization}
\usetikzlibrary{datavisualization.formats.functions}
\usepackage{wrapfig}
\usepackage{xcolor}
\usepackage{amsmath}
\usepackage{mathtools}
\usepackage{suffix}
\usepackage{amsfonts}
\usepackage{mathpartir}
\usepackage{enumitem}
\usepackage{stmaryrd}
\usepackage[all]{xy}
\usepackage{amssymb}
\usepackage{twoopt}

\makeatletter
\newcommand\@TyAlph[1]{%
\ifcase #1\or \tau\or \sigma\or \rho\else \@ctrerr \fi%
}
\newcommand\ty[1][1]{{\@TyAlph{#1}}}

\newcommand\tvar[1][1]{{\@TyVarAlph{#1}}}
\newcommand\@TyVarAlph[1]{%
\ifcase #1\or \alpha\or \beta\or \gamma\else \@ctrerr \fi%
}

\newcommand\var[1][1]{{\@VarAlph{#1}}}
\newcommand\@VarAlph[1]{%
\ifcase #1\or x\or y\or z\or u\or v\or w\else \@ctrerr \fi%
}

\newcommand\trm[1][1]{{\@TermAlph{#1}}}
\newcommand\@TermAlph[1]{%
\ifcase #1\or t\or s\or r\else \@ctrerr \fi%
}

\newcommand\val[1][1]{%
\ifcase #1\or v\or w\or u\else \@ctrerr \fi%
}

\newcommand\op[1][1]{%
\ifcase #1\or \mathsf{op}\or \mathsf{op}'\or \mathsf{op}''\else \@ctrerr \fi%
}
\makeatother

\newcommand\Op{\mathsf{Op}}
\newcommand\cnst{\underline{c}}
\newcommand\sigmoid{\varsigma}

\newcommand\bp[1]{\boldsymbol{(}#1\boldsymbol{)}}

\newcommand\tPair[2]{\langle #1, #2\rangle}
\newcommand\tTriple[3]{\langle #1, #2, #3\rangle}
\newcommand\tTuple[1]{\langle #1\rangle}
\newcommand\tInj[3][\,]{#2.#3#1}
\newcommand\Cns{\ell}
\newcommand\Inj[2][\,]{\mathsf{#2}#1}
\newcommand\tNothingSym{\mathsf{Nothing}}

\newcommand\tJustSym{\mathsf{Just}}

\newcommand\tNil{[\,]}
\newcommand\tCons[2]{#1::#2}

\newcommand\fun[1]{\lambda #1.}
\newcommand\letin[3]{\mathbf{let}\,#1=\,#2\,\mathbf{in}\,#3}

\newcommand\pMatch[5][\,]{\mathbf{case}\,#2\,\mathbf{of}#1\tPair{#3}{#4}\To#5}
\newcommand\tMatch[4][\,]{\mathbf{case}\,#2\,\mathbf{of}#1\tTuple{#3}\To#4}
\newcommand\vMatch[3][\,]{\mathbf{case}\,#2\,\mathbf{of}#1\{#3\}}

\newcommand\lFold[5]{\mathbf{fold}\, (#1,#2).#3\,\mathbf{over}\,#4\,\mathbf{from}\,#5}

\newcommand\ctx{\Gamma}
\newcommand\tinf{\vdash}
\newcommand\Ginf[3][]{\ctx #1\tinf #2 : #3}
\newcommand\subst[2]{#1{}[#2]}
\newcommand\sfor[2]{^{#2}\!/\!_{#1}}

\newcommand\reals{\mathbf{real}}

\newcommand\Unit{\bp{\,}}
\newcommand\Variant[1]{\{ #1 \}}
\WithSuffix\newcommand\t*{\boldsymbol{\mathop{*}}}
\newcommand\MaybeSym{\mathbf{maybe}}
\newcommand\ListSym{\mathbf{list}}

\newcommand\Maybe[1]{\MaybeSym(#1)}
\newcommand\List[1]{\ListSym(#1)}

\newcommand\To{\to}

\newcommand\bProd[2]{\bp{#1 \t* #2}}
\newcommand\tProd[3]{\bp{#1 \t* #2 \t* #3}}

\newcommand\Dsynsymbol[1][]{\scalebox{0.8}{$\overrightarrow{\mathcal{D}}$}_{#1}}
\newcommand\Dsyn[2][]{\Dsynsymbol[#1](#2)}

\newcommand\Dsynrevsymbol[1][]{\scalebox{0.8}{$\overleftarrow{\mathcal{D}}$}_{#1}}
\newcommand\Dsynrev[2][]{\Dsynrevsymbol[#1](#2)}

\newcommand\Syn{\mathbf{Syn}}

\newcommand\tFromMaybe[1]{\mathrm{fromMaybe}}
\newcommand\tFromMayben[2]{\mathrm{fromMaybe}^{#2}}

\newcommand\tMap[2]{\mathrm{map}}

\newcommand{\plots}[1]{\mathcal{P}_{#1}}

\newcommand\freeeq[1]{\stackrel{\# #1}{=}}

\definecolor{shade}{RGB}{223,223,223}
\definecolor{unshade}{RGB}{255,255,255}

\usepackage{tcolorbox}
\newtcbox{\shadebox}{on line,arc=1pt, outer arc=2pt,%
  colback=shade,colframe=shade,boxsep=0pt,%
  left=1pt,right=1pt,top=2pt,bottom=2pt,%
  boxrule=0pt,bottomrule=1pt,toprule=1pt}

\newtcbox{\unshadebox}{on line,arc=1pt, outer arc=2pt,%
  colback=unshade,colframe=shade,boxsep=0pt,%
  left=1pt,right=1pt,top=2pt,bottom=2pt,%
  boxrule=0pt,bottomrule=1pt,toprule=1pt}

\newcommand\syncat[1]{\mspace{-25mu}\synname{#1}}
\newcommand\synname[1]{\qquad\text{#1}}

\newenvironment{syntax}[1][]{%
\(
  \rowcolors{100}{white}{white}%
  \begin{array}[t]{#1l@{\quad\!\!}*3{l@{}}@{\,}l}
}{
\end{array}
\)%
}

\newcommand\gdefinedby{::=}
\newcommand\gor{\mathrel{\lvert}}

\newcommand\vor{\mathrel{\big\lvert}}

\newcommand{\dif}{\mathop{}\!\mathrm{d}}
\newcommand{\Diff}{\mathbf{Diff}}
\newcommand{\CartSp}{\mathbf{CartSp}}
\newcommand{\Man}{\mathbf{Man}}

\newcommand{\Set}{\mathbf{Set}}
\newcommand{\sem}[1]{\llbracket #1\rrbracket}
\newcommand{\semgl}[1]{\llparenthesis #1\rrparenthesis}
\newcommand{\RR}{\mathbb{R}}
\newcommand{\NN}{\mathbb{N}}
\newcommand\catC{\mathcal{C}}

\newcommand\freeF{F}

\newcommand\Dsemsymbol[1][]{\mathcal{T}^{#1}}
\newcommand\Dsem[2][]{\Dsemsymbol[#1](#2)}

\newcommand\evRsymbol[1][]{\mathrm{evR}^{#1}}
\newcommandtwoopt\evR[3][][]{\evRsymbol[#2]_{#1}(#3)}
\newcommand\lamRsymbol[1][]{\mathrm{lamR}^{#1}}
\newcommandtwoopt\lamR[3][][]{\lamRsymbol[#2]_{#1}(#3)}

\newcommand{\Gl}[1][]{\mathbf{Gl}_{#1}}

\newcommand{\sPair}[2]{( #1, #2 )}
\newcommand{\sTriple}[3]{(#1, #2, #3)}

\newcommand\id[1][{}]{{\rm id}_{#1}}

\newcommand\projf{\mathrm{proj}}

\newcommand\xto\xrightarrow

\newcommand\DtoT[2][]{\phi_{#2#1}^{\Dsynsymbol\Dsemsymbol}}

\newcommand\innerprod[1]{\cdot_{#1}}

\newcommand\y{\mathbf{y}}

\newcommand\seq[2][]{\left(#2\right)_{#1}}
\newcommand\coseq[2][]{\left[#2\right]_{#1}}

\newcommand\cover{\mathcal{U}}

\renewcommand\lim{\mathrm{lim}}

\renewcommand\circ{
  \PackageError{Matthijs}{Please use ; instead! Or if you'd rather use circ consistently, tell me and we'll switch the whole document over}{}}

\newcommand{\defeq}{\stackrel {\mathrm{def}}=}

\usepackage{todonotes}

\newenvironment{tightitemize}{\begin{itemize}[leftmargin=*,noitemsep,topsep=0pt]}{\end{itemize}\ignorespacesafterend}

\begin{document}
\title{Correctness of Automatic Differentiation via Diffeologies and Categorical Gluing
}
\titlerunning{Correctness of AD via Diffeologies and Categorical Gluing}
%
\author{Mathieu Huot \Letter\inst{1}\inst{*}
\and
Sam Staton\inst{1}\inst{*}
\and
Matthijs V\'ak\'ar\inst{2}\inst{*}
}
\authorrunning{M. Huot, S. Staton, M. V\'ak\'ar}
%
\institute{University of Oxford, UK \and
Utrecht University, The Netherlands \\
\email{${}^*$\textnormal{Equal contribution}\qquad{mathieu.huot@stx.ox.ac.uk}}
}
%
\maketitle              

\begin{abstract}
We present semantic correctness proofs of Automatic Differentiation (AD).
We consider a forward-mode AD method on a higher order language with algebraic data types, and we characterise it as the unique structure preserving macro given a choice of derivatives for basic operations.
We describe a rich semantics for differentiable programming, based on diffeological spaces.
We show that it interprets our language, and we phrase what it means for the AD method to be correct with respect to this semantics.
We show that our characterisation of AD gives rise to an elegant semantic proof of its correctness based on a gluing construction on diffeological spaces.
We explain how this is, in essence, a logical relations argument.
Finally, we sketch how the analysis extends to other AD methods by considering a continuation-based method.

\end{abstract}
%
%
%

\section{Introduction}
Automatic differentiation (AD), loosely speaking, is the process of taking a program describing a function, and building the derivative of that function by applying the chain rule across the program code.
As gradients play a central role in many aspects of machine learning, so too do automatic differentiation systems such as TensorFlow~\cite{abadi2016tensorflow} or Stan~\cite{carpenter2015stan}.

\begin{wrapfigure}[9]{r}{73mm} \vspace{-15mm}
\[
	\xymatrix@C+2mm{
          *+[F]{\txt{Programs}}
          \ar[d]_-{\txt{\footnotesize denotational\\\footnotesize semantics}} \ar[rr]^{\txt{\footnotesize automatic\\\footnotesize differentiation}}
		 && *+[F]{\txt{Programs}} \ar[d]^-{\txt{\footnotesize denotational\\\footnotesize semantics}}\\
		 *+[F]{\txt{Differential\\geometry}}\ar[rr]^{\txt{\footnotesize math\\\footnotesize differentiation}}
		 && *+[F]{\txt{Differential\\geometry}}
	}
  \]\vspace{-5mm}\caption{Overview of semantics/correctness of AD.\label{fig:intro}}\end{wrapfigure}
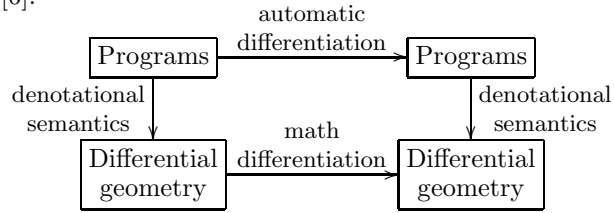
 Differentiation has a well developed mathematical theory in terms of differential geometry. The aim of this paper is
to formalize this connection between differential geometry
and the syntactic operations of AD. 
In this way we achieve two things: (1)~a compositional, denotational understanding of differentiable programming and AD; (2)~an explanation of the correctness of AD.

This intuitive correspondence (summarized in Fig.~\ref{fig:intro}) is in fact rather complicated.
  In this paper we focus on resolving the following problem: higher order functions play a key role in programming, and yet they have no counterpart in traditional differential geometry. Moreover, we resolve this problem while retaining the compositionality of denotational semantics.

  \newpage
  
\vspace{-10pt}\subsubsection{Higher order functions and differentiation.}
A major application of higher order functions is to support disciplined code reuse.
Code reuse is particularly acute in machine learning.
For example,
a multi-layer neural network might be built of millions of near-identical neurons,
as follows. 
\newcommand{\neuron}{\mathrm{neuron}}
\newcommand{\layer}{\mathrm{layer}}
\newcommand{\compose}{\mathrm{comp}}
\[\begin{array}{ll}\begin{aligned}
  &\neuron_n:\bProd{\reals^n}{\bProd{\reals^n}{\reals}}\To\reals
  \\&\neuron_n\defeq \lambda \tTuple{x,\tTuple{w,b}}.\,\sigmoid( w\cdot x+b) 
  \\
  &\layer_n:(\bProd{\ty_1}{P}\To \ty_2)\To \bProd{\ty_1}{P^n}\To \ty_2^n
  \\
  &\layer_n\defeq\lambda f.\,\lambda\tTuple{x,\tTuple{p_1,\dots, p_n}}.\,
  \tTuple{f\tTuple{x,p_1},\dots,f\tTuple{x,p_n}}
  \\
  &\compose : \bProd{(\bProd{\ty_1}{P}\To \ty_2)}{(\bProd{\ty_2}{Q}\To \ty_3)}\To \bProd{\ty_1}{\bProd{P}{Q}}\To \ty_3
    \\
  &\compose\defeq\lambda \tTuple{f,g}.\,\lambda\tTuple{x,(p,q)}.\,g\tTuple{f\tTuple{x,p},q}
\end{aligned}
    &\hspace{-16mm}
\raisebox{-8mm}[0pt]{\begin{tikzpicture}[xscale=0.5,yscale=0.7]
      \datavisualization [scientific axes=clean,
                    y axis=grid,
                    visualize as smooth line,
                    y axis={label={$\sigmoid(x)$},ticks={step=0.5}},
                    x axis={label, ticks={
                                        step = 5}} ]
data [format=function] {
      var x : interval [-9:9] samples 100;
      func y = 1/(1 + exp(-\value x));
      };
      \end{tikzpicture}}
\end{array}\]
(Here $\sigmoid(x) \defeq\frac 1 {1+e^{-x}}$ is the sigmoid function, as illustrated.)
We can use these functions to build a network as follows (see also Fig.~\ref{fig:network}):
\begin{equation}\label{eqn:network}
  \compose\tTuple{\layer_m(\neuron_k),\compose\tTuple{\layer_n(\neuron_m),\neuron_n}}
  :\bProd{\reals^k}{P}\to\reals
\end{equation}
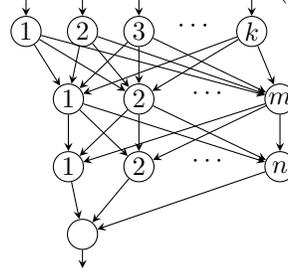
\begin{wrapfigure}[12]{r}{40mm} \vspace{-10mm}
 \tikzset{%
  every neuron/.style={
    circle,
    draw,
    minimum size=4mm
  },
  neuron missing/.style={
    draw=none, 
    scale=1,
    text height=0.01cm,
    execute at begin node=\color{black}$\cdots$
  },
}

\begin{tikzpicture}[x=1.5cm, y=1.5cm, >=stealth, rotate=90,xscale=-0.3,yscale=0.5]

\foreach \m/\l [count=\y] in {1,2,3,missing,4}
  \node [every neuron/.try, neuron \m/.try] (input-\m) at (0,2.5-\y) {};

\foreach \m [count=\y] in {1,2,missing,3}
  \node [every neuron/.try, neuron \m/.try ] (hidden-\m) at (2,2-\y*1.25) {};

\foreach \m [count=\y] in {1,2,missing,3}
  \node [every neuron/.try, neuron \m/.try ] (hiddenb-\m) at (4,2-\y*1.25) {};

\foreach \m [count=\y] in {1}
  \node [every neuron/.try, neuron \m/.try ] (output-\m) at (6,1.5-\y) {};

\foreach \l [count=\i] in {1,2,3,k}
  \draw [<-] (input-\i) -- ++(-1,0)
    node at (input-\i) {$\l$};

\foreach \l [count=\i] in {1,2,m}
  \node at (hidden-\i) {$\l$};

\foreach \l [count=\i] in {1,2,n}
  \node at (hiddenb-\i) {$\l$};

\foreach \l [count=\i] in {1}
  \draw [->] (output-\i) -- ++(1,0)
    node [above, midway] {};

\foreach \i in {1,...,4}
  \foreach \j in {1,...,3}
    \draw [->] (input-\i) -- (hidden-\j);

\foreach \i in {1,...,3}
  \foreach \j in {1,...,3}
    \draw [->] (hidden-\i) -- (hiddenb-\j);

\foreach \i in {1,...,3}
  \foreach \j in {1}
    \draw [->] (hiddenb-\i) -- (output-\j);


\end{tikzpicture}

\vspace{-3mm}\caption{The network in~\eqref{eqn:network} with $k$ inputs and two hidden layers.\label{fig:network}}\end{wrapfigure}
Here $P\cong \reals^p$ with $p=(m(k{+}1){+}n(m{+}1){+}n{+}1)$. 
This program~\eqref{eqn:network} describes a smooth (infinitely differentiable) function.
The goal of automatic differentiation is to find its derivative.

If we $\beta$-reduce all the $\lambda$'s, we end up with a very long function expression just built from the sigmoid function and linear algebra. We can then find a program for calculating its derivative by applying the chain rule.
However, automatic differentiation can also be expressed without first $\beta$-reducing,
in a compositional way, by explaining how higher order functions like $(\layer)$ and $(\compose)$  propagate derivatives. This paper is a semantic analysis of this compositional approach.

The general idea of denotational semantics is to interpret types as spaces and programs as functions between the spaces. In this paper, we propose to use
diffeological spaces and smooth functions~\cite{souriau1980groupes,iglesias2013diffeology} to this end.
These satisfy the following three desiderata:
\begin{itemize}
\item $\RR$ is a space, and the smooth functions $\RR\to\RR$ are exactly the functions that are infinitely differentiable;
\item The set of smooth functions $X\to Y$ between spaces again forms a space,
  so we can interpret function types.
\item The disjoint union of a sequence of spaces again forms a space, so we can interpret variant types and inductive types. 
\end{itemize}
We emphasise that the most standard formulation of differential geometry, using manifolds, does not support spaces of functions. Diffeological spaces seem to us the simplest notion of space that satisfies these conditions, but there are other candidates~\cite{baez2011convenient,smootheology}.
A diffeological space is in particular a set $X$ equipped with a chosen set of curves
$C_X\subseteq X^\RR$
and a smooth map $f:X\to Y$ must be such that if $\gamma\in C_X$ then
$\gamma;f\in C_Y$. 
This is remiscent of the method of logical relations.

\vspace{-10pt}\subsubsection{From smoothness to automatic derivatives at higher types.}
Our denotational semantics in diffeological spaces
guarantees that all definable functions are smooth.
But we need more than just to know that a definable function happens to have a mathematical derivative: we need to be able to find that derivative.

In this paper we focus on a simple, forward mode automatic differentiation method, which is a macro translation on syntax (called~$\Dsynsymbol$ in \S\ref{sec:simple-language}). 
We are able to show that it is correct, using our denotational semantics.

Here there is one subtle point that is central to our development.
Although differential geometry provides established derivatives for first order functions
(such as $\neuron$ above),
there is no canonical notion of derivative for higher order functions (such as $\layer$ and $\compose$) 
in the theory of diffeological spaces (e.g.~\cite{christensen2014tangent}). 
We propose a new way to resolve this, by interpreting types as triples $(X,X',S)$ where, intuitively, $X$ is a space of inhabitants of the type, $X'$ is a space serving as a chosen bundle of tangents over $X$, and $S\subseteq X^\RR\times X'^\RR$ is a binary relation between curves, informally relating curves in $X$ with their tangent curves in $X'$.
This new model gives a denotational semantics for automatic differentiation. 

In \S\ref{sec:semantics} we boil this new approach down to a straightforward and elementary logical relations argument for the correctness of automatic differentiation. The approach is explained in detail in \S\ref{sec:correctness}.

\vspace{-10pt}\subsubsection{Related work and context.}
AD has a long history and has many implementations.
AD was perhaps first phrased in a functional setting in~\cite{pearlmutter2008reverse}, and there are now a number of teams working on AD in the functional setting 
(e.g.~\cite{wang2018demystifying,shaikhha2019efficient,elliott2018simple}), some providing efficient implementations. 
Although that work does not involve formal semantics, it is inspired by intuitions from differential geometry and category theory. 

This paper adds to a very recent body of work on verified automatic differentiation. Much of this is concurrent with and independent from the work in this article. 
In the first order setting,
there are recent accounts based on denotational semantics in manifolds~\cite{fong2019backprop} and based on synthetic differential geometry~\cite{gallagher-sdg},
as well as work making a categorical abstraction~\cite{rev-deriv-cat2020} and
work connecting operational semantics 
with denotational semantics \cite{abadi-plotkin2020,plotkin-invited-talk}.
Recently there has also been significant progress at higher types. The work of Brunel et al.~gives formal correctness proofs for reverse-mode derivatives on computation graphs~\cite{brunel2019backpropagation}.
The work of Barthe et al.~\cite{bcdg-open-logical-relations} provides a general discussion of some new syntactic logical relations arguments including one very similar to our syntactic proof of Theorem~\ref{thm:fwd-cor-basic}. 
We understand that the authors of~\cite{gallagher-sdg} are working on higher types.

The differential $\lambda$-calculus \cite{ehrhard2003differential} is related to AD, and explicit connections are made in \cite{mak-ong2020,Manzyuk2012}. One difference is that the differential $\lambda$-calculus allows addition of terms at all types, and hence vector space models are suitable, but this appears peculiar with the variant and inductive types that we consider here. 

Finally we emphasise that we have chosen the neural network~(\ref{eqn:network})
as our running example mainly for its simplicity. There are many other examples of AD outside the neural networks literature:
AD is useful whenever derivatives need to be calculated on high dimensional spaces. This includes optimization problems more generally, where the derivative is passed to a
gradient descent method (e.g.~\cite{robbins1951stochastic,kiefer1952stochastic,qian1999momentum,kingma2014adam,duchi2011adaptive,liu1989limited}).
Other applications of AD are in advanced \emph{integration} methods, since derivatives play a role in 
Hamiltonian Monte Carlo~\cite{neal2011mcmc,hoffman2014no} and variational inference~\cite{kucukelbir2017automatic}. 

\vspace{-10pt}\subsubsection{Summary of contributions.}
We have provided a semantic analysis of automatic differentiation. 
Our syntactic starting point is a well-known forward-mode AD macro on a typed higher order language (e.g.~\cite{shaikhha2019efficient,wang2018demystifying}). We recall this in \S\ref{sec:simple-language} 
for function types, and in \S\ref{sec:extended-language} we extend it to inductive types and variants. 
 The main contributions of this paper are as follows. 
\begin{itemize}
\item We give a denotational semantics for the language in diffeological spaces, showing that every definable expression is smooth (\S\ref{sec:semantics}).
\item We show correctness of the AD macro by a logical relations argument (Th.~\ref{thm:fwd-cor-basic}).
\item We give a categorical analysis of this correctness argument with two parts: canonicity of the macro in terms of syntactic categories, and a new notion of glued space that abstracts the logical relation (\S\ref{sec:correctness}).
\item We then use this analysis to state and prove a correctness argument at all first order types (Th.~\ref{thm:fwd-cor-full}). 
\item We show that our method is not specific to one particular AD macro, by also considering a continuation-based AD method~(\S\ref{sec:rev-mode-short}). 
\end{itemize}




\section{A simple forward-mode AD translation}\label{sec:simple-language}
\subsubsection{Rudiments of differentiation and dual numbers.}
Recall that the derivative of a function $f:\RR\to \RR$, if it exists, is a function
$\nabla f:\RR\to \RR$ such that $\nabla f(x_0)=\frac {\dif f(x)}{\dif x}(x_0)$ is the gradient of $f$ at $x_0$. 

To find $\nabla f$ in a compositional way, two generalizations are reasonable:
\begin{tightitemize}
\item We need both $f$ and $\nabla f$ when calculating $\nabla (f;g)$
of a composition $f;g$, using the chain rule, so we are really interested in the pair $(f,\nabla f):\RR\to \RR\times \RR$;
\item In building $f$ we will need to consider functions of multiple arguments, such as $+:\RR^2\to \RR$, and these functions should propagate derivatives.
\end{tightitemize}
Thus we are more generally interested in transforming a function $g:\RR^n\to \RR$ into a function
$h:(\RR\times \RR)^n\to \RR\times \RR$ in such a way that for any
$f_1\dots f_n:\RR\to\RR$, 
\begin{equation}
  \label{eqn:dualnumber}
  (f_1,\nabla f_1,\dots, f_n,\nabla f_n);h
  =
  ((f_1,\dots, f_n);g,\nabla ((f_1, \dots, f_n);g))\text.
\end{equation}

An intuition for $h$ is often given in terms of dual numbers.
The transformed function operates on pairs of numbers, $(x,x')$, and it is common
to think of such a pair as $x+x'\epsilon$ for an `infinitesimal' $\epsilon$.
But while this is a helpful intuition, the formalization of infinitesimals can be intricate, 
and the development in this paper is focussed on the elementary formulation in~\eqref{eqn:dualnumber}. 

The reader may also notice that $h$ encodes all the partial derivatives of
$g$. For example, 
if $g \colon \RR^2\to \RR$, then with $f_1(x)\defeq x$ and $f_2(x)\defeq x_2$, by applying \eqref{eqn:dualnumber} to $x_1$ we obtain
$h(x_1,1,x_2,0)=(g(x_1,x_2), \frac {\partial g(x,x_2)}{\partial x}(x_1))$
and similarly 
$h(x_1,0,x_2,1)=(g(x_1,x_2), \frac {\partial g(x_1,x)}{\partial x}(x_2))$.
And conversely, if $g$ is differentiable in each argument, then
a unique $h$ satisfying \eqref{eqn:dualnumber} can be found by taking linear
combinations of partial derivatives:
\[\textstyle h(x_1,x_1',x_2,x_2')=(g(x_1,x_2),x_1' \cdot\frac {\partial g(x,x_2)}{\partial x}(x_1)+x_2'\cdot \frac {\partial g(x_1,x)}{\partial x}(x_2))\text.\]

In summary, the idea of differentiation with dual numbers is 
to transform a differentiable function
$g:\RR^n\to \RR$ to a function $h:\RR^{2n}\to \RR^2$ which captures~$g$ and all its partial derivatives. We packaged this up in~\eqref{eqn:dualnumber} as a sort-of invariant which is useful for building derivatives of compound functions $\RR\to\RR$ in a compositional way.
The idea of forward mode automatic differentiation is to perform this transformation at the source code level. 

\vspace{-10pt}
\subsubsection{A simple language of smooth functions.}
We consider a standard higher order typed language with a first order type $\reals$ of real numbers. The types $(\ty,\ty[2])$ and terms $(\trm,\trm[2])$ are as follows.

\noindent\begin{syntax}
  \ty, \ty[2], \ty[3] & \gdefinedby & & \syncat{types}                          \\
  &\gor& \reals                      & \synname{real numbers}\\
\end{syntax}%
~
\begin{syntax}
  &\gor\quad\,& \tProd{\ty_1}{\dots}{\ty_n} & \synname{finite product} \\
&\gor& \ty \To \ty[2]              & \synname{function}      \\
\end{syntax}

\noindent\begin{syntax}
    \trm, \trm[2], \trm[3] & \gdefinedby & & \syncat{terms}                          \\
    &    & \var                          & \synname{variable}                        \\
    &\gor& \cnst
    \ \gor\ \trm+\trm[2]
    \ \gor\ \trm*\trm[2]
    \ \gor\ \sigmoid(\trm)
    & \synname{operations/constants}                      \\
    &\gor& \tTriple{\trm_1}{\dots}{\trm_n}
    \ \gor  \tMatch{\trm}{\var_1,\dots, \var_n}{\trm[2]}\hspace{-10pt} \;& \synname{tuples/pattern matching}\\
    &\gor& \fun \var    \trm
    \ \gor  \trm\, \trm[2]               & \synname{function abstraction/app.}\\

\end{syntax}

\noindent The typing rules are in Figure~\ref{fig:types1}. We have included a minimal set of operations for the sake of illustration, but it is not difficult to add further operations. We add some simple syntactic sugar $t-u\defeq
t+\underline{(-1)}* u$. We intend $\sigmoid$ to stand for the sigmoid function,
$\sigmoid(x)\defeq\frac 1 {1+e^{-x}}$.
We further include syntactic sugar $\letin{\var}{\trm}{\trm[2]}$ for $(\fun{\var}{\trm[2]})\,\trm$
and $\fun{\tTuple{\var_1,\ldots,\var_n}}{\trm}$ for $\fun{\var}{\tMatch{\var}{\var_1,\ldots,\var_n}{\trm}}$.

\begin{figure}[b]
  \framebox{\scalebox{0.82}{\begin{minipage}{1.2\linewidth}\noindent\input{type-system}\end{minipage}}}
  \caption{Typing rules for the simple language.\label{fig:types1}}
  \end{figure}

\subsubsection{Syntactic automatic differentiation: a functorial macro.}
The aim of forward mode AD is to find the dual numbers representation of a function by syntactic manipulations.
For our simple language, we implement this as the following inductively defined macro $\Dsynsymbol$ on both types and terms
(see also~\cite{wang2018demystifying,shaikhha2019efficient}):

\noindent\[
\begin{array}{ll}
\Dsyn{\reals} \defeq \bProd{\reals}{\reals} &
\Dsyn{\ty\To\ty[2]} \defeq \Dsyn{\ty}\To\Dsyn{\ty[2]} \\
\Dsyn{\tProd{\ty_1}{\cdots}{\ty_n}} \defeq \tProd{\Dsyn{\ty_1}}{\cdots}{\Dsyn{\ty_n}}\hspace{40pt}\, &\vspace{-3pt}
\end{array}    
\]
\noindent\[
\begin{array}{l}
\Dsyn{\var} \defeq \var\hspace{80pt}\Dsyn{\cnst} \defeq \tPair{\cnst}{0}\\
\Dsyn{\trm+\trm[2]} \defeq \pMatch{\Dsyn{\trm}}{\var}{\var'}
                       {\pMatch{\Dsyn{\trm[2]}}{\var[2]}{\var[2]'}
                       {\tPair{\var + \var[2]}{\var' + \var[2]'}}}\\
\Dsyn{\trm*\trm[2]} \defeq  \pMatch{\Dsyn{\trm}}{\var}{\var'}
                        {\pMatch{\Dsyn{\trm[2]}}{\var[2]}{\var[2]'}
                        {\tPair{\var * \var[2]}{\var * \var[2]' + \var' * \var[2]}}}\\
\Dsyn{\sigmoid(\trm)} \defeq \pMatch{\Dsyn{\trm}}{\var}{\var'}{
\letin{\var[2]}{\sigmoid(\var)}{
\tPair{\var[2]}{\var'*\var[2]*(1-\var[2])}
}    
}\\
\Dsyn{\fun \var    \trm} \defeq \fun\var{\Dsyn{\trm}}\hspace{12pt}
\Dsyn{\trm\, \trm[2] } \defeq 
\Dsyn{\trm}\,\Dsyn{\trm[2]}\hspace{12pt}
\Dsyn{\tTriple{\trm_1}{\dots}{\trm_n}} \defeq \tTriple{\Dsyn{\trm_1}}{\dots}{\Dsyn{\trm_n}} \\
\Dsyn{{\tMatch{\trm}{\var_1,\dots,\var_n}{\trm[2]}}} \defeq
\tMatch{\Dsyn\trm}{\var_1,\dots,\var_n }{\Dsyn{\trm[2]}} \\
\end{array}
\]

\vspace{-5pt}
We extend $\Dsynsymbol$ to contexts: $\Dsyn{\{\var_1{:}\ty_1,{.}{.}{.},\var_n{:}\ty_n\}}\defeq
\{\var_1{:}\Dsyn{\ty_1},{.}{.}{.},\var_n{:}\Dsyn{\ty_n}\}$.
This turns $\Dsynsymbol$ into a well-typed, functorial macro in the following sense.
\begin{lemma}[Functorial macro]
	If $\ctx\tinf \trm:\ty$ then $\Dsyn{\ctx}\tinf \Dsyn{\trm}:\Dsyn{\ty}$.\\
	If $\ctx,\var:\ty[2]\tinf \trm:\ty$ and
	$\ctx\tinf\trm[2]:\ty[2]$ then
	$\Dsyn{\ctx}\tinf \Dsyn{\subst{\trm[1]}{\sfor{\var}{\trm[2]}}}=\subst{\Dsyn{\trm}}{\sfor{\var}{\Dsyn{\trm[2]}}
	}$.
\end{lemma}

\vspace{-5pt}
\begin{example}[Inner products]\label{ex:innerprod}
Let us write $\ty^n$ for the $n$-fold product $\tProd{\ty}{\dots}{\ty}$.
Then, given $\Ginf{\trm,\trm[2]}{\reals^n}$ we can define their inner product\\
$
\begin{array}{ll}
\Gamma\vdash\trm\innerprod{n}\trm[2]\defeq\;&\tMatch{\trm}{\var[3]_1,\ldots,\var[3]_n}{}\\
&\tMatch{\trm[2]}{\var[2]_1,\ldots,\var[2]_n}{\var[3]_1 * \var[2]_1 + \dots + \var[3]_n * \var[2]_n}
:\reals	
\end{array}
$\\
To illustrate the calculation of $\Dsynsymbol$, let us expand (and $\beta$-reduce) $\Dsyn{\trm\innerprod{2}\trm[2]}$:\\
$
\pMatch{\Dsyn{\trm}}{\var[3]_1}{\var[3]_2}{}\pMatch{\Dsyn{\trm[2]}}{\var[2]_1}{\var[2]_2}{}\pMatch{\var[3]_1}{\var[3]_{1,1}}{\var[3]_{1,2}}{}\\
\pMatch{\var[2]_1}{\var[2]_{1,1}}{\var[2]_{1,2}}{} 
\pMatch{\var[3]_2}{\var[3]_{2,1}}{\var[3]_{2,2}}{}
\pMatch{\var[2]_2}{\var[2]_{2,1}}{\var[2]_{2,2}}{}\\
\tPair{\var[3]_{1,1}*\var[2]_{1,1}+\var[3]_{2,1}*\var[2]_{2,1}\ }{\ \var[3]_{1,1}*\var[2]_{1,2}+\var[3]_{1,2}*\var[2]_{1,1}+\var[3]_{2,1}*\var[2]_{2,2}+\var[3]_{2,2}*\var[2]_{2,1}}
$
\end{example}

\begin{example}[Neural networks]
In our introduction~\eqref{eqn:network}, we provided a program in our language 
to build a neural network out of expressions $\neuron,\layer,\compose$;
this program makes use of the inner product of Ex.~\ref{ex:innerprod}.
We can similarly calculate $\Dsynsymbol$ of such deep neural nets by mechanically applying the macro.
\end{example}


\section{Semantics of differentiation}\label{sec:semantics}
Consider for a moment the first order fragment of the language in \S~\ref{sec:simple-language}, with only one type, $\reals$, 
and no $\lambda$'s or pairs. 
This has a simple semantics in the category of cartesian spaces and smooth maps.
Indeed, a term $\var_1\dots\var_n:\reals \vdash \trm:\reals$ has a natural reading
as a function $\sem{\trm}:\RR^n\to\RR$
by interpreting our operation symbols by the 
well-known operations on $\RR^n\to\RR$ with the corresponding name.
In fact, the functions that are definable in this first order fragment are smooth, 
which means that they are continuous, differentiable, and their derivatives are continuous, differentiable, and so on. 
Let us write $\CartSp$ for this category of cartesian spaces ($\RR^n$ for some $n$)
and smooth functions.

The category $\CartSp$ has cartesian products, and so we can also interpret product types, tupling and pattern matching,
giving us a useful syntax
for constructing functions into and out of products of $\RR$.
For example, the interpretation of $(\neuron_n)$ in (\ref{eqn:network})
becomes\vspace{-8pt}
\[
\RR^n\times \RR^n\times \RR \xto{\sem{\innerprod{n}}\times \id[\RR]}\RR\times \RR\xto{\sem{+}}\RR\xto{\sem{\sigmoid}}\RR.
\vspace{-8pt}
\]
where $\sem{\innerprod{n}}$, $\sem{+}$ and $\sem{\sigmoid}$ are the usual inner product, addition
and the sigmoid function on $\RR$, respectively.

Inside this category, we can straightforwardly study the first order language without $\lambda$'s, and automatic differentiation.
In fact, we can prove the following by plain induction on the syntax:\\
\emph{The interpretation of the (syntactic) forward AD $\Dsyn{\trm}$ of a first-order term
$\trm$ equals the usual (semantic) derivative of the interpretation of $\trm$ as a smooth function.}

However, as is well known, the category $\CartSp$ does not support function spaces. To see this, 
notice that we have polynomial terms \vspace{-5pt}
\[\var_1,\ldots,\var_d:\reals\vdash \lambda \var[2].\,\textstyle\sum_{n=1}^d \var_n\var[2]^n:\reals\to\reals\vspace{-5pt}\]
for each $d$, and so if we could interpret $(\reals\to \reals)$ as a Euclidean space 
$\RR^p$ then, by interpreting these polynomial expressions, we would 
be able to find continuous injections $\RR^d\to \RR^p$ for every $d$, which is topologically impossible for any~$p$, for example as a consequence of the 
Borsuk-Ulam theorem (see \ifx\fossacsversion\undefined Appx.~\ref{sec:man_not_ccc}\else\cite{hsv-fossacs2020}, Appx.~A\fi).

This means that we cannot interpret the functions $(\layer)$ and $(\compose)$ from~(\ref{eqn:network}) in $\CartSp$, as they are higher order functions,
even though they are very useful and innocent building blocks for differential programming!
Clearly, we could define neural nets such as~(\ref{eqn:network}) directly as smooth functions 
without any higher order subcomponents, though that would quickly become cumbersome for deep networks.
A problematic consequence of the lack of a semantics for higher order differential programs is that we have no obvious way of establishing compositional semantic correctness of $\Dsynsymbol$ for the given implementation of~(\ref{eqn:network}).

\vspace{-10pt}
\subsubsection{Diffeological spaces.}
This motivates us to turn to a more general notion of differential geometry for our semantics, based on 
\emph{diffeological spaces} \cite{iglesias2013diffeology}.
The key idea will be that a higher order function
is called smooth if it sends smooth functions to smooth functions, meaning that we can never use it
to build first order functions that are not smooth. For example, $(\compose)$ in~(\ref{eqn:network}) has this property.
\begin{definition}
	A \emph{diffeological space} $(X,\plots{X})$ consists of a set $X$ together with, for each $n$ and each open subset $U$ of $\RR^n$,  a set $\plots{X}^U\subseteq [U\to X]$ of functions, called \emph{plots}, such that
	\begin{tightitemize}
	 	\item all constant functions are plots;
	 	\item if $f:V\to U$ is a smooth function and $p\in\plots{X}^U$, then $f;p\in\plots{X}^V$;
     \item if $\seq[i\in I]{p_i\in\plots{X}^{U_i}}$ is a compatible family of plots $(x\in U_i\cap U_j\Rightarrow p_i(x)=p_j(x))$
     and $\seq[i\in I]{U_i}$ covers $U$,
     then the gluing $p:U\to X:x\in U_i\mapsto p_i(x)$ is a plot.
	 \end{tightitemize} 
\end{definition}
We call a function $f:X\to Y$ between diffeological spaces \emph{smooth} if, for all plots
$p\in\plots{X}^U$, we have that $p;f\in \plots{Y}^U$. We write $\Diff(X,Y)$ for the set of smooth maps from $X$ to $Y$. 
Smooth functions compose, and so we have a category $\Diff$ of diffeological spaces and smooth functions.

A diffeological space is thus a set equipped with structure.
Many constructions of sets carry over straightforwardly to diffeological spaces.

\begin{example}[Cartesian diffeologies]\label{ex:cartesian-diffeologies}
Each open subset $U$ of $\RR^n$ can be given the structure of a diffeological space by taking all the
smooth functions $V\to U$ as $\plots{U}^V$.
It is easily seen that smooth functions from $V\to U$ in the traditional sense coincide with
smooth functions in the sense of diffeological spaces.
Thus diffeological spaces have a profound relationship with ordinary calculus.
\end{example}
 In categorical terms, this gives a full embedding of $\CartSp$ in $\Diff$. 
\begin{example}[Product diffeologies]
Given a family $\seq[i\in I]{X_i}$ of diffeological spaces,
we can equip the product $\prod_{i\in I}X_i$ of sets with the
\emph{product diffeology} in which $U$-plots are precisely the functions
of the form $\seq[i\in I]{p_i}$ for $p_i\in\plots{X_i}^U$.  
\end{example}
This gives us the categorical product in $\Diff$.
\begin{example}[Functional diffeology]
We can equip the set $\Diff(X,Y)$ of smooth functions between diffeological spaces with the \emph{functional diffeology}
in which $U$-plots consist of functions $f:U\to \Diff(X,Y)$ such that 
$(u,x)\mapsto f(u)(x)$ is an element of $\Diff(U\times X, Y)$.
\end{example}
This specifies the categorical function object in $\Diff$.

\vspace{-10pt}
\subsubsection{Semantics and correctness of AD.}
We can now give a denotational semantics to our language from \S~\ref{sec:simple-language}. 
We interpret each type $\ty$ as a set $\sem \ty$ equipped with the relevant diffeology,
by induction on the structure of types:
\vspace{-7pt}
\[
\sem \reals \defeq \RR\qquad
\sem{\tProd{\ty_1}{\dots}{\ty_n}}\ \defeq\ \textstyle\prod_{i=1}^n\sem{\ty_i}
\qquad
\sem{\ty\To\ty[2]} \defeq \Diff(\sem \ty,\sem{\ty[2]})
\vspace{-7pt}
\]
A context $\Gamma=(\var_1\colon\ty_1\dots \var_n\colon \ty_n)$ is interpreted as a diffeological space
$\sem \Gamma\defeq \prod_{i=1}^n\sem{\ty_i}$. 
Now well typed terms $\Gamma\vdash \trm:\ty$ are interpreted as smooth functions
$\sem \trm:\sem\Gamma\to \sem \ty$, giving a meaning for $\trm$ for every valuation of the context. 
This is routinely defined by induction on the structure of typing derivations. 
Constants $\cnst:\reals$ are interpreted as constant functions;
and the first order operations ($+,*,\sigmoid$) are interpreted by composing with the corresponding functions, which are smooth.
For example, $\sem{\sigmoid(t)}(\rho)\defeq\sigmoid(\sem t(\rho))$, where $\rho\in\sem\Gamma$. 
Variables are interpreted as $\sem{\var_i}(\rho) \defeq \rho_i$. 
The remaining constructs are interpreted as follows, and it is straightforward to show that smoothness is preserved. \vspace{-6pt}
\begin{align*}
&
\sem{\tTriple {\trm_1}{\dots}{ \trm_n}}(\rho)\defeq
(\sem{\trm_1}(\rho),\dots,\sem{\trm_n}(\rho))
&&
\sem{\fun{\var{:}\ty}{\trm}}(\rho)(a)\defeq
\sem {\trm}(\rho,a)\ \text{($a\in \sem {\ty}$)}
\\
&
\sem{\tMatch{\trm}{{.}{.}{.} }{\trm[2]}}(\rho)\defeq
\sem{\trm[2]}(\rho,\sem{\trm}(\rho))
&&
\sem{\trm\,\trm[2]}(\rho)\defeq
\sem{\trm}(\rho)(\sem{\trm[2]}(\rho))\vspace{-6pt}
\end{align*}

Notice that a term
$\var_1\colon\reals,\dots,\var_n\colon\reals\vdash \trm : \reals$
is interpreted as a smooth function
$\sem \trm:\RR^n\to \RR$, even if $t$ involves higher order functions (like~(\ref{eqn:network})).
Moreover the macro differentiation $\Dsyn\trm$ is a function 
$\sem {\Dsyn \trm}:(\RR\times \RR)^n\to (\RR\times \RR)$. 
This enables us to state a limited version of our main correctness theorem:
\begin{theorem}[Semantic correctness of $\Dsynsymbol$ (limited)]
  \label{thm:fwd-cor-basic}
  For any term\\ $\var_1\colon\reals,\dots,\var_n\colon\reals\vdash \trm : \reals$, the function
  $\sem {\Dsyn \trm}$ is the dual numbers representation \eqref{eqn:dualnumber} of
  $\sem \trm$.
  In detail: for any smooth functions
  $f_1\dots f_n:\RR\to\RR$, 
  \[(f_1,\nabla f_1,\dots, f_n,\nabla f_n);\sem{\Dsyn\trm}
    =
    \big((f_1\dots f_n);\sem\trm,\nabla ((f_1\dots f_n);\sem{\trm})\big)\text.
  \]
\end{theorem}
(For instance, if $n=2$, then 
$\sem{\Dsyn\trm}(\var_1,1,\var_2,0)=
(\sem\trm(\var_1,\var_2),\frac{\partial\sem\trm(\var,\var_2)}{\partial \var}(\var_1))$.)
\begin{proof}
  We prove this by logical relations.  
  Although the following proof is elementary, we found it by
  using the categorical methods in \S~\ref{sec:correctness}.

For each type $\ty$, we define a binary relation
 $S_{\ty}$ between curves in $\sem{\ty}$ and curves in $\sem{\Dsyn{\ty}}$,
 i.e.~$S_{\ty}\subseteq  \plots{\sem{\ty}}^\RR\times
 \plots{\sem{\Dsyn{\ty}}}^\RR$,
 by induction on $\ty$:
\begin{tightitemize}
\item $S_{\reals}\defeq\{(f,(f,\nabla f))~|~f:\RR\to\RR\text{ smooth}\}$;
\item $S_{\bProd{\ty}{\ty[2]}} \defeq \{\sPair{\sPair{f_1}{g_1}}{\sPair{f_2}{g_2}}\mid (f_1,f_2)\in S_{\ty}, (g_1,g_2)\in S_{\ty[2]}\}$;
\item $S_{\ty\To \ty[2]} \defeq \{(f_1,f_2) \mid \forall (g_1,g_2) \in S_{\ty}. (x{\mapsto} f_1(x)(g_1(x)),x{\mapsto} f_2(x)(g_2(x))) \in S_{\ty[2]} \}$.
\end{tightitemize}

\noindent Then, we establish the following `fundamental lemma':

\begin{quotation}
\noindent If $\var_1{:} \ty_1,{.}{.}{.}, \var_n {:} \ty_n \vdash \trm : \ty[2]$
and, for all $1{\leq} i{\leq} n$, for all smooth
$f_i: \RR\to \sem{\ty_i}$ and $g_i : \RR\to \sem{\Dsyn{\ty_i}}$ such that $(f_i,g_i)$ is in $S_{\ty_i}$,
we have that 
$\left((f_1,\ldots,f_n);\sem{\trm},(g_1,\ldots,g_n);\sem{\Dsyn{\trm}}\right)$
is in $S_{\ty[2]}$.
\end{quotation}


This is proved routinely by induction on the typing derivation of $\trm$.
The case for $*$ relies on the precise definition of $\Dsyn{\trm*\trm[2]}$,
and similarly for $+,\sigmoid$. 

We conclude the theorem from the fundamental lemma by considering the case where $\ty_i=\ty[2]=\reals$,
$m=n$ and $s_i=y_i$. \qed
\end{proof}

\section{Extending the language: variant and inductive types}\label{sec:extended-language}
In this section, we show that the definition of forward AD and the semantics generalize
if we extend the language of \S \ref{sec:simple-language} with variants
and inductive types.
As an example of inductive types, we consider lists.
This specific choice is only for expository purposes and the whole
development works at the level of generality of arbitrary algebraic
data types generated as initial algebras of (polynomial) type constructors formed by
finite products and variants.

Similarly, our choice of operations is for expository purposes. More generally, assume given a family of operations $(\Op_n)_{n\in\NN}$ indexed by their arity $n$. Further assume that each $\op\in\Op_n$ has type $\reals^n\to \reals$. We then ask for a certain closure of these operations under differentiation, that is we define\\
$
\begin{array}{ll}
	\Dsyn{\op(\trm_1,\ldots,\trm_n)}\defeq~ &\pMatch{\Dsyn{\trm_1}}{\var_1}{\var_1'}
                       { \ldots \to\pMatch{\Dsyn{\trm_n}}{\var_n}{\var_n'}
                       {\\ &\tPair{\op(\var_1,\ldots,\var_n)}{\sum_{i=1}^n\var_i' *\partial_i\op(\var_1,\ldots,\var_n)}}}
\end{array}
$\\
where $\partial_i\op(\var_1,\ldots,\var_n)$ is some chosen term in the language, involving free variables from $\var_1,\ldots,\var_n$, which we think of as implementing the partial derivative of $\op$ with respect to its $i$-th argument.
For constructing the semantics, every $\op$ must be interpreted by some smooth function, and, to establish correctness, the semantics 
of $\partial_i\op(\var_1,\ldots,\var_n)$ must be the semantic $i$-th partial derivative of the semantics of
$\op(\var_1,\ldots,\var_n)$.

\vspace{-10pt}
\subsubsection{Language.}\label{sec:extended-language-language}
We additionally consider the following types and terms:

\noindent\begin{syntax}
    \ty, \ty[2], \ty[3] & \gdefinedby & & \syncat{types}                          \\
    &\gor& \Variant{
        \Inj{\Cns_1}{\ty_1}
        \vor \ldots \vor
        \Inj{\Cns_n}{\ty_n}
      }        &\synname{variant}          \\
\end{syntax}%
~
\begin{syntax}
&\gor\quad\,& \List{\ty}                 & \synname{list}\\
\end{syntax}

\noindent\begin{syntax}
    \trm, \trm[2], \trm[3] & \gdefinedby & & \syncat{terms}                          \\
    &\gor&\tInj\ty\Cns\trm               & \synname{variant constructor}             \\
    &\gor& \tNil
    \ \gor\ \tCons{\trm}{\trm[2]}          & \synname{empty list and cons}\\
    &\gor& \vMatch {\trm  }  {
                        \Inj{\Cns_1}{\var_1}\To{\trm[2]_1}
                \vor \cdots
                \vor  \Inj{\Cns_n}{\var_n}\To{\trm[2]_n}
                }           & \synname{pattern matching: variants}\\
    &\gor& \lFold{\var_1}{\var_2}{\trm}{\trm[2]}{\trm[3]} & \synname{list fold}\\

\end{syntax}

We extend the type system according to:

\noindent\[
\begin{array}{@{}c@{}}
  \inferrule{
    \Ginf\trm{\ty_i}
  }{
    \Ginf{\tInj\ty{\Cns_i}\trm}{\ty}
  }((\Inj{\Cns_i}\ty_i) \in \ty)
\quad
  \inferrule{
    ~
  }{
    \Ginf \tNil {\List{\ty}} 
  }
  \quad
  \inferrule{
  \Ginf \trm \ty
  \\
  \Ginf {\trm[2]} {\List{\ty}}
  }{
  \Ginf {\tCons{\trm}{\trm[2]}} {\List{\ty}}
  }
\\
  \inferrule{
    \Ginf\trm{\Variant{
                \Inj{\Cns_1}{\ty_1}
                \vor \ldots \vor
                \Inj{\Cns_n}{\ty_n}}}
    \\
    \text{for each $1 \leq i \leq n$: }
    \Ginf[, \var_i : \ty_i]{\trm[2]_i}{\ty}
  }{
    \Ginf{\vMatch \trm
                {\begin{array}[t]{@{}l@{\,}l@{}l@{}}
                    \Inj{\Cns_1}{\var_1}\To{\trm[2]_1}
                    \vor\cdots
                    \vor\Inj{\Cns_n}{\var_n}&\To{\trm[2]_n}
    }}
    \ty
                  \end{array}
  }
  \\
  \inferrule{
    \Ginf {\trm[2]} {\List{\ty}}
    \\
    \Ginf {\trm[3]} {\ty[2]}
    \\
    \Ginf[{,\var_1:\ty,\var_2:\ty[2]}] {\trm} {\ty[2]}
    }{
    \Ginf {\lFold{\var_1}{\var_2}{\trm}{\trm[2]}{\trm[3]}} {\ty[2]}
    }
\end{array}
\]

We can then extend $\Dsynsymbol$ to our new types and terms by

\noindent\[
\begin{array}{lr}
\Dsyn{\Variant{\Inj{\Cns_1}{\ty_1}\vor \ldots \vor\Inj{\Cns_n}{\ty_n}}} \defeq 
\Variant{\Inj{\Cns_1}{\Dsyn{\ty_1}}\vor \ldots \vor\Inj{\Cns_n}{\Dsyn{\ty_n}}}\qquad\, &
\Dsyn{\List{\ty}} \defeq \List{\Dsyn{\ty}}
\end{array}    
\]

\noindent\[
\begin{array}{l}
\Dsyn{\tInj\ty\Cns\trm} \defeq \tInj{\Dsyn{\ty}}\Cns{\Dsyn\trm} \qquad\qquad\quad
\Dsyn{\tNil} \defeq \tNil\qquad\quad\qquad
\Dsyn{\tCons{\trm}{\trm[2]}} \defeq \tCons{\Dsyn{\trm}}{\Dsyn{\trm[2]}}\\
\Dsyn{\vMatch {\trm  }  {
    \Inj{\Cns_1}{\var_1}\To{\trm[2]_1}
\vor \cdots
\vor  \Inj{\Cns_n}{\var_n}\To{\trm[2]_n}
}} \defeq\\
\qquad\vMatch {\Dsyn\trm  }  {
    \Inj{\Cns_1}{\var_1}\To{\Dsyn{\trm[2]_1}}
\vor \cdots
\vor  \Inj{\Cns_n}{\var_n}\To{\Dsyn{\trm[2]_n}}
} \\
\Dsyn{\lFold{\var_1}{\var_2}{\trm}{\trm[2]}{\trm[3]}} \defeq
\lFold{\var_1}{\var_2}{\Dsyn\trm}{\Dsyn{\trm[2]}}{\Dsyn{\trm[3]}}
\end{array}
\]

To demonstrate the practical use of expressive type systems for
differential programming, we consider the following two examples.
\begin{example}[Lists of inputs for neural nets]
Usually, we run a neural network on a large data set, the size of 
which might be determined at runtime.
To evaluate a neural network on multiple inputs, in practice, one often sums the outcomes.
This can be coded in our extended language as follows.
Suppose that we have a network $f:\bProd{\reals^n}{P}\To\reals$ that operates on single input vectors. 
We can construct one that operates on lists of inputs as follows:
\[
g\defeq \fun{\tTuple{l,w}}{\lFold{\var_1}{\var_2}{f\tTuple{\var_1,w} + \var_2}{l}{\underline{0}}} :
\bProd{\List{\reals^n}}{P}\To\reals
\]
\end{example}

\begin{example}[Missing data]
In practically every application of statistics and machine learning,
we face the problem of \emph{missing data}:
for some observations, only partial information is available.
In an expressive typed programming language like we consider,
we can model missing data conveniently using the data type
$\Maybe{\ty}=\Variant{
  \Inj{\tNothingSym}{\Unit}\vor
  \Inj{\tJustSym}{\ty}
} $.
In the context of a neural network, one might use it as follows.
First, define some helper functions
\[
\begin{aligned}
&\tFromMaybe{\reals} \defeq\fun{\var}{\fun{m}{
  \vMatch {m  }  {
    \Inj{\tNothingSym}{\_}\To{\var}
\vor  \Inj{\tJustSym}{\var'}\To{\var'}
}  
}}
\\&
\tFromMayben{\ty}{n}\defeq
\fun{\tTriple{\var_1}{{.}{.}{.}}{\var_n}}{\fun{\tTriple{m_1}{{.}{.}{.}}{m_n}}{\tTriple{\tFromMaybe{\ty}\,\var_1\,m_1}{{.}{.}{.}}{\tFromMaybe{\ty}\,\var_n\,m_n}}  }\\&\qquad\qquad:(\Maybe{\reals})^n\To \reals^n\To\reals^n\\&
\tMap{\ty}{\ty[2]}\defeq\fun{f}{\fun{l}{\lFold{\var_1}{\var_2}{\tCons{f\, \var_1}{\var_2}}{l}{\tNil}}}
: (\ty\To\ty[2])\To\List{\ty}\To\List{\ty[2]}
\end{aligned}
\]
Given a neural network $f:\bProd{\List{\reals^k}}{P}\To\reals$,
we can build a new one that operates on 
on a data set for which some covariates (features) are missing, by passing in 
default values to replace the missing covariates:
\begin{multline*}
\fun{\tTuple{l,\tTuple{m,w}}}
f\tTuple{\tMap{\reals}{\reals}\, (\tFromMayben{\reals}{k}\,m)\, l
,w}
\\:\bProd{\List{(\Maybe{\reals})^k}}{\bProd{\reals^{k}}{P}}\To\reals\end{multline*}
Then, given a data set $l$ with missing covariates, we can perform automatic differentiation on this network to optimize, simultaneously, the ordinary network parameters~$w$ \emph{and} the default values for missing covariates $m$.
\end{example}

\subsubsection{Semantics.}\label{sec:extended-language-semantics}
In \S~\ref{sec:semantics} we gave a denotational semantics for the simple language in diffeological spaces. This extends to the language in this section, as follows.
As before, each type $\ty$ is interpreted as a diffeological space, which is a set equipped with a family of plots:
\begin{tightitemize}
\item A variant type $\Variant{
        \Inj{\Cns_1}{\ty_1}
        \vor \ldots \vor
        \Inj{\Cns_n}{\ty_n}
      } $ is inductively interpreted as the disjoint union
      $\textstyle\sem{\Variant{
        \Inj{\Cns_1}{\ty_1}
        \vor \dots \vor
        \Inj{\Cns_n}{\ty_n}
      }} \ \ \defeq \ \ \biguplus_{i=1}^n\sem {\ty_i}$ with $U$-plots\\
    $\plots{
    \sem{\Variant{
        \Inj{\Cns_1}{\ty_1}
        \vor \ldots \vor
        \Inj{\Cns_n}{\ty_n}
      }}
    }^U\hspace{-4pt}\defeq
    \left\{\coseq[j= 1]{U_j\xto {f_j}\sem{\ty_j}\to\biguplus_{i=1}^n\sem{\ty_i}}^n\hspace{-4pt}~\big|~U=\biguplus_{j=1}^n U_j,\;f_j\in\plots{\sem{\ty_j}}^{U_j}\right\}$.
  \item A list type $\List\ty$ is interpreted as the set of lists,
    $ \sem{\List\ty} \ \ \defeq\ \  \biguplus_{i=1}^\infty \sem\ty^i$
    with $U$-plots\\
    $\plots{
      \sem{\List\ty}
      }^U\defeq\left\{\coseq[j=1]{U_j\xto {f_j}\sem{\ty}^j\to\biguplus_{i=1}^\infty\sem{\ty}^i}^\infty~\big|~U=\biguplus_{j=1}^\infty U_j,\;f_j\in\plots{\sem{\ty}^j}^{U_j}\right\}$.
  \end{tightitemize}
  The constructors and destructors for variants and lists are interpreted as
  in the usual set theoretic semantics.
  It is routine to show inductively that these interpretations are smooth. Thus every term
  $\Gamma\vdash \trm:\ty$ in the extended language is interpreted as a smooth function
  $\sem \trm:\sem\Gamma\to\sem \ty$ between diffeological spaces. 
  

(In this section we focused on a language with lists, but other inductive types are easily interpreted in the category of diffeological spaces in much the same way; the categorically minded reader may regard this as a consequence of $\Diff$ being a concrete Grothendieck quasitopos, e.g.~\cite{baez2011convenient}.)


\section{Categorical analysis of forward AD and its correctness}\label{sec:correctness}
This section has three parts. First, we give a categorical account of the functoriality of AD (Ex.~\ref{ex:canonical-fwd}). Then we introduce our gluing construction, and relate it to the correctness of AD~(dgm.~(\ref{dgm:gluing})).
Finally, we state and prove a correctness theorem for all first order types by considering a category of manifolds~(Th.~\ref{thm:fwd-cor-full}). 
\vspace{-5pt}
\subsubsection{Syntactic categories.}
\begin{figure}[b]
  \framebox{\scalebox{0.8}{\begin{minipage}{1.22\linewidth}\vspace{-3mm}\input{beta-eta-all}
\end{minipage}}}
\caption{Standard $\beta\eta$-laws (e.g.~\cite{pitts1995categorical}) for products, functions, variants and lists. \label{fig:beta-eta}}
\end{figure}
Our language induces a syntactic category as follows.
\begin{definition}
	Let $\Syn$ be the category whose objects are types,
        and where a morphism $\ty[1]\to\ty[2]$ is a term in context 
        $\var:\ty[1]\vdash \trm:\ty[2]$ modulo the $\beta\eta$-laws
        (Fig.~\ref{fig:beta-eta}).
        Composition is by substitution. 
\end{definition}
For simplicity, we do not impose arithmetic identities such as $x+y=y+x$ in $\Syn$.
As is standard, this category has the following universal property.
\begin{lemma}[e.g.~\cite{pitts1995categorical}]\label{lem:syn-initial}
  For every bicartesian closed category $\catC$ with list objects,
  and every object $\freeF(\reals)\in\catC$ and
  morphisms $\freeF(\cnst)\in\catC(1, \freeF(\reals))$,
  $\freeF(+), \freeF(*)\in\catC(\freeF(\reals)\times \freeF(\reals),\freeF(\reals))$, $\freeF(\sigmoid)\in\Syn(\freeF(\reals),\freeF(\reals))$
  in $\catC$, there is a unique functor $\freeF:{\Syn\to\catC}$ respecting the interpretation and preserving the bicartesian closed structure as well as
 list objects.
\end{lemma}
\begin{proof}[notes]
  The functor $\freeF:\Syn\to \catC$ is a canonical denotational semantics
  for the language, interpreting types as objects of $\catC$ and terms as morphisms. 
  For instance,
  $\freeF({\ty\To\ty[2]})\defeq (\freeF\ty\To\freeF{\ty[2]})$,
  the function space in the category $\catC$,
  and $\freeF{(\trm\,\trm[2])}\defeq
  $ is the composite $(\freeF{\trm},\freeF{\trm[2]});\mathit{eval}$. 
  When $\catC=\Diff$,
the denotational semantics of the language in diffeological spaces (\S\ref{sec:semantics},\ref{sec:extended-language-semantics})
can be understood as the unique structure preserving functor 
$\sem-:\Syn\to \Diff$ satisfying 
$\sem \reals=\RR$, $\sem \sigmoid=\sigmoid$ and so on. 
  \qed
\end{proof}
\begin{example}[Canonical definition forward AD]\label{ex:canonical-fwd}
The forward AD macro $\Dsynsymbol$ (\S\ref{sec:simple-language},\ref{sec:extended-language-language}) 
arises as a canonical cartesian closed functor on $\Syn$. 
Consider the unique  
cartesian closed functor $\freeF:\Syn\to\Syn$ 
such that $\freeF(\reals)=\reals* \reals$,
$\freeF(\cnst)=\Dsyn\cnst$, 
$\freeF(\sigmoid)=\Dsyn{\sigmoid(x)}$, 
and\\
$
\begin{array}{l}
\freeF(+) = \var[3]:\freeF(\reals)\t* \freeF(\reals)\vdash 
  \tMatch{\var[3]}{\var,\var[2]}{\Dsyn{\var+\var[2]}}:\freeF(\reals)\\
\freeF(*)=  \var[3]:\freeF(\reals)\t* \freeF(\reals)\vdash 
  \tMatch{\var[3]}{\var,\var[2]}{\Dsyn{\var*\var[2]}}:\freeF(\reals)
\end{array}
$\\
  Then for any type $\ty$, $F(\ty)=\Dsyn \ty$, and 
  for any term $x:\ty[1]\vdash \trm:\ty[2]$, 
  $F(\trm)=\Dsyn \trm$ as morphisms $F(\ty[1])\to F(\ty[2])$ in the syntactic category. \vspace{-10pt}
\end{example}

\subsubsection{Categorical gluing and logical relations.}
Gluing is a method for building new categorical models which has been used for many purposes, including logical relations and realizability~\cite{mitchell1992notes}. Our logical relations argument in the proof of Th.~\ref{thm:fwd-cor-basic} can be understood in this setting. In this subsection, for the categorically minded, we explain this, and in doing so we quickly recover a correctness result for the more general language in \S~\ref{sec:extended-language} and 
for arbitrary first order types.

We define a category $\Gl[U]$ whose objects are triples $(X,X',S)$ where $X$ and~$X'$ are diffeological spaces and $S\subseteq\plots X^U\times \plots {X'}^U$ is a relation between their $U$-plots. A morphism $(X,X',S)\to (Y,Y',T)$ is a pair of smooth functions $f\colon X\to Y$, $f'\colon X'\to Y'$, such that if 
$(g,g')\in S$ then $(g;f,g';f')\in T$.
The idea is that this is a semantic domain where we can simultaneously interpret the language and its automatic derivatives.
\begin{proposition}\label{prop:gluing}
  The category $\Gl[U]$ is bicartesian closed, has list objects, and the projection functor 
  $\projf:\Gl[U]\to\Diff\times \Diff$ preserves this structure. 
\end{proposition}
\begin{proof}[notes]
The category $\Gl[U]$ is a full subcategory of the comma category $\id[\Set]\downarrow \Diff(U,-)\times \Diff(U,-)$. 
The result thus follows by the general theory of categorical gluing~(e.g.~\cite[Lemma~15]{johnstone-lack-sobocinski}). 
\qed
\end{proof}
We give a semantics $\semgl{-}=(\semgl{-}_0,\semgl{-}_1, S_{-})$ for the language in $\Gl[\RR]$, 
interpreting types $\ty$ as objects $(\semgl{ \ty}_0,\semgl{\ty}_1,S_{\ty})$,
and terms as morphisms. 
We let $\semgl{\reals}_0\defeq \RR$ and 
$\semgl{\reals}_1\defeq \RR^2$, with the relation
$S_\reals\defeq \{(f,(f,\nabla f))~|~f:\RR\to\RR \text{ smooth}\}$.
We interpret the constants $\cnst$ as pairs $\semgl{ \cnst}_0\defeq \cnst$ and
$\semgl{\cnst}_1\defeq (\cnst,0)$,
and we interpret $+,\times,\sigmoid$ in the standard way (meaning, like $\sem{-}$) in $\semgl{-}_0$, but according to the derivatives in $\semgl{-}_1$, for instance,
$\semgl{*}_1:\RR^2\times\RR^2\to \RR^2$ is 
\[\semgl{*}_1((x,x'),(y,y'))\defeq (xy,xy'+x'y)\text.\]
At this point one checks that these interpretations are indeed morphisms in $\Gl[\RR]$. 
This amounts to checking that these interpretations are dual numbers 
representations in the sense of~(\ref{eqn:dualnumber}). 
The remaining constructions of the language are interpreted using the categorical structure of $\Gl[\RR]$, following Lem.~\ref{lem:syn-initial}.

Notice that the diagram below commutes. One can check this by hand or note that it follows from the initiality of $\Syn$ (Lem.~\ref{lem:syn-initial}):
all the functors preserve all the structure. 
\vspace{-6pt}
\begin{equation}\label{dgm:gluing}
\xymatrix{
\Syn\ar[rr]^-{(\id,\Dsyn-)}\ar[d]_{\semgl-}&&\Syn\times \Syn\ar[d]^{\sem-\times \sem-}
\\
\Gl[\RR]\ar[rr]_-{\projf}&&\Diff\times\Diff
}
\end{equation}
We thus arrive at a restatement of the correctness theorem (Th.~\ref{thm:fwd-cor-basic}), which holds even for the extended language with variants and lists, because 
for any $x_1{.}{.}{.}x_n:\reals\vdash \trm:\reals$, 
the interpretations $(\sem \trm,\sem{\Dsyn \trm})$ are in the image of the projection $\Gl[\RR]\to \Diff\times\Diff$, 
and hence $\sem{\Dsyn\trm}$ is a dual numbers encoding~of~$\sem\trm$.

\subsubsection{Correctness at all first order types, via manifolds.}
We now generalize Theorem~\ref{thm:fwd-cor-basic} to hold at all first order types, not just the reals. To do this, we need to define the derivative of a smooth map between the interpretations of first order types.
We do this by recalling the well known theory of manifolds and tangent bundles. 

For our purposes, a smooth manifold $M$ is a second-countable Hausdorff
topological space together with a smooth atlas: an open cover $\cover$ together
with homeomorphisms $\seq[U\in\cover]{\phi_U:U\to \RR^{n(U)}}$ (called charts) such that $\phi_U^{-1};\phi_V$ is smooth on its domain of definition for all $U,V\in\cover$.
A function $f:M\to N$ between manifolds is smooth if $\phi^{-1}_U;f;\psi_V$ is
smooth for all charts $\phi_U$ and $\psi_V$ of $M$ and $N$, respectively.
Let us write $\Man$ for this category.

Our manifolds are slightly unusual because different charts in an atlas may have different finite dimension $n(U)$. Thus we consider manifolds with dimensions that are potentially unbounded, albeit locally finite. This does not affect the theory of differential geometry as far as we need it here.


Each open subset of $\RR^n$ can be regarded as a manifold. This lets us regard the category of manifolds $\Man$ as a full subcategory of the category of diffeological spaces. We consider a manifold $(X,\{\phi_U\}_U)$ as a diffeological space with the same carrier set $X$ and where the plots $\plots X ^U$ are the smooth functions in $\Man(U,X)$. A function $X\to Y$ is smooth in the sense of manifolds if and only if it is smooth in the sense of diffeological spaces~\cite{iglesias2013diffeology}. For the categorically minded reader, this means that we have a full embedding of $\Man$ into $\Diff$.
Moreover, the natural interpretation of the first order fragment of our language in $\Man$ coincides with that in $\Diff$.
That is, the embedding of $\Man$ into $\Diff$ preserves finite products and countable coproducts (hence initial algebras of polynomial endofunctors).
\begin{proposition}
  Suppose that a type $\ty$ is first order, i.e.~it is just built from reals, products, variants, and lists (or, again, arbitrary inductive types), and not function types. Then the diffeological space $\sem{\ty}$ is a manifold. 
\end{proposition}
\begin{proof}[notes] This is proved by induction on the structure of types. In fact, one may show that every such $\sem\ty$ is isomorphic to a manifold of the form $\biguplus_{i=1}^n\RR^{d_i}$ where the bound $n$ is either finite or $\infty$, but this isomorphism is typically not an identity function. \qed
\end{proof}
The constraint to first order types is necessary because, e.g. the space $\sem{\reals\to\reals}$ is not a manifold, because of a Borsuk-Ulam argument (see Appx. \ref{sec:man_not_ccc}).

We recall that the derivative of any morphism $f:M\to N$ of manifolds 
is given as follows.
For each point $x$ in a manifold $M$, define the \emph{tangent space} $\Dsemsymbol_x M$ to be the set $\{\gamma\in\Man(\RR,M)\mid \gamma(0)=x\}/\sim$ of equivalence classes $[\gamma]$ of smooth curves $\gamma$
in $M$ based at $x$, where we identify $\gamma_1\sim \gamma_2$ iff $\nabla (\gamma_1;f)(0)=\nabla(\gamma_2;f)(0)$ for all smooth $f:M\to\RR$.
The \emph{tangent bundle} of $M$ is the set $\Dsem{M}\defeq \biguplus_{x\in M} \Dsemsymbol_x (M)$. The charts of $M$ equip $\Dsem{M}$ with a canonical manifold structure.
Then for smooth $f:M\to N$, the derivative $\Dsem{f}:\Dsem{M}\to\Dsem{N}$ is defined as $\Dsem{f}\sPair{x}{[\gamma]}\defeq \sPair{f(x)}{[\gamma;f]}$.
All told, the derivative is a functor $\Dsemsymbol:\Man\to\Man$. 

As is standard, we can understand the tangent bundle of a composite space in terms of that of its parts.
\begin{lemma}\label{lemma:Dsem}
There are canonical isomorphisms $\Dsem{\biguplus_{i=1}^\infty M_i} \cong \biguplus_{i=1}^\infty \Dsem{M_i}$ and 
$\Dsem{M_1\times\ldots \times M_n}\cong \Dsem{M_1}\times\ldots\times \Dsem{M_n}$.
\end{lemma}
We define a canonical isomorphism $\DtoT{\ty}:\sem{\Dsyn{\ty}}\to\Dsem{\sem{\ty}}$ for every type $\ty$, by induction on the structure of types. We let $\DtoT{\reals}:\sem{\Dsyn{\reals}}\to\Dsem{\sem{\reals}}$ be given by $\DtoT{\reals}(x,x')\defeq (x,[t\mapsto x+x't])$.
For the other types, we use Lemma~\ref{lemma:Dsem}. 
We can now phrase correctness at all first order types.
 \begin{theorem}[Semantic correctness of $\Dsynsymbol$ (full)]\label{thm:fwd-cor-full}
   For any ground $\ty$, any first order context $\Gamma$
   and any term $\Gamma\vdash\trm:\ty$,
   the syntactic translation $\Dsynsymbol$ coincides with the tangent bundle functor, modulo these canonical isomorphisms:
   \[\xymatrix{
       \sem{\Dsyn\Gamma}\ar[r]^{\sem{\Dsyn \trm}}\ar[d]_{\DtoT{\Gamma}}^\cong
       &\sem{\Dsyn \ty}\ar[d]^{\DtoT{\ty}}_\cong
       \\
       \Dsem{\sem\Gamma}\ar[r]_{\Dsem{\sem \trm}}
       &\Dsem {\sem \ty}
       }\]
 \end{theorem}
 \begin{proof}[notes]
   For any curve $\gamma\in\Man(\RR,M)$, let $\bar \gamma\in\Man(\RR,\Dsem M)$ be the tangent curve, given by $\bar \gamma(x)=(\gamma(x),[t\mapsto \gamma(x+t)])$.
   First, we note that a smooth map $h:\Dsem M\to \Dsem N$ is of the form $\Dsem{g}$ for some $g:M\to N$ if for all smooth curves $\gamma:\RR\to M$ we have $\bar \gamma;h=\overline{(\gamma;g)}:\RR\to \Dsem N$. This generalizes~(\ref{eqn:dualnumber}).
   Second, for any first order type $\ty$, $S_{\sem\ty}=\{(f,\tilde{f})~|~\tilde{f};\DtoT\tau=\bar f\}$. This is shown by induction on the structure of types. 
   We conclude the theorem from diagram (\ref{dgm:gluing}), by putting these two observations together.\qed

 \end{proof}


\section{A continuation-based AD algorithm}\label{sec:rev-mode-short}
\vspace{-6pt}
We now illustrate the flexibility of our framework by briefly describing an alternative syntactic translation $\Dsynrevsymbol[\rho]$. This alternative translation uses aspects of continuation passing style, inspired by recent developments in reverse mode AD~\cite{wang2018demystifying,brunel2019backpropagation}.
In brief, $\Dsynrevsymbol[\rho]$ works by
$\Dsynrev[\rho]\reals=(\reals*(\reals\To\rho))$. Thus instead of using a pair of a number and its tangent, we use a pair of a number and a continuation. The answer type $\rho=\reals^k$ needs to have the structure of a vector space, and the continuations that we consider will turn out to be linear maps. Because we work in continuation passing style, the chain rule is applied contravariantly.
If the reader is familiar with reverse-mode AD algorithms, they may think of the dimension $k$ as the number of memory cells used to store the result. 

Computing the whole gradient of a term $\var_1:\reals,{.}{.}{.},\var_k:\reals\vdash \trm:\reals$ at once is then achieved
by running $\Dsynrev[k]{\trm}$ on a $k$-tuple of basis vectors for $\reals^k$.

We define the continuation-based AD macro
$\Dsynrevsymbol[k]$ on types and terms as the unique structure preserving functor
$\Syn\to\Syn$
with $
    \Dsynrev[k]{\reals} = \bProd{\reals}{(\reals\To\reals^k)}
$~and\\
$\Dsynrev[k]{\cnst} \defeq \tPair{\cnst}{\fun{\var[3]}{\tTriple{\underline{0}}{\ldots}{\underline{0}}}}$\\
$
\begin{array}{l}
    \Dsynrev[k]{\trm+\trm[2]} \defeq \pMatch{\Dsynrev[k]{\trm}}{\var}{\var'}
    {\pMatch{\Dsynrev[k]{\trm[2]}}{\var[2]}{\var[2]'}
    {\tPair{\var + \var[2]}{\fun{\var[3]}{\var'\, \var[3] + \var[2]'\, \var[3]}}}}\\
    \Dsynrev[k]{\trm*\trm[2]} \defeq  \pMatch{\Dsynrev[k]{\trm}}{\var}{\var'}
    {\pMatch{\Dsynrev[k]{\trm[2]}}{\var[2]}{\var[2]'}{}}\\
    \hspace{165pt}\phantom{\Dsynrev[k]{\trm*\trm[2]\defeq}}{{\tPair{\var * \var[2]}{\fun{\var[3]}{\var[1]'\,(\var[2] * \var[3]) + \var[2]'\,(\var * \var[3])}}}}\\ 
    \Dsynrev[k]{\sigmoid(\trm)} \defeq \pMatch{\Dsynrev[k]{\trm}}{\var}{\var'}{
\letin{\var[2]}{\sigmoid(\var)}{
\tPair{\var[2]}{\fun{\var[3]}{\var'\,(\var[2]*(1-\var[2])*\var[3])}}
}    
}.
\end{array}$\\
Here, we use sugar $\var:\reals^k,\var[2]:\reals^k\vdash \var+\var[2]\defeq 
\tMatch{\var}{\var_1,\ldots,\var_k}{}
\tMatch{\var[2]}{\var[2]_1,\ldots,\var[2]_k}{}
\tTriple{\var_1+\var[2]_1}{\ldots}{\var_k+\var[2]_k}$.
(We could easily expand this definition by making $\Dsynrevsymbol[k]$ preserve all
other term and type formers, as we did for $\Dsynsymbol$.)
Note that the corresponding scheme for an arbitrary $n$-ary operation $\op$ would be (c.f. the
scheme for forward AD in \S\ref{sec:extended-language})
\\
$
\begin{array}{ll}
	\Dsynrev[k]{\op(\trm_1,\ldots,\trm_n)}\defeq\hspace{-1pt} &\pMatch{\Dsynrev[k]{\trm_1}}{\var_1}{\var_1'}
                       { \ldots \to\pMatch{\Dsynrev[k]{\trm_n}}{\var_n}{\var_n'}
                       {\\ &\tPair{\op(\var_1,\ldots,\var_n)}{\fun{\var[3]}{\sum_{i=1}^n\var_i' (\partial_i\op(\var_1,\ldots,\var_n)*\var[3])}}}}.
                    \end{array}
                    $\\
The idea is that $\Dsynrev[k]{\trm}$ is a higher order function that simultaneously computes
$\trm$ (the forward pass) and defines as a continuation the reverse pass which computes the gradient.
In order to actually run the algorithm, we need two auxiliary definitions\\
$
\begin{array}{l}
  \lamRsymbol[k]_\reals\defeq \lambda{\var[3]}.\,
    \pMatch{\var[3]}{\var}{\var'}{}
    \tMatch{\var'}{\var'_1,\ldots,\var'_k}{}\\
    \phantom{  \lamRsymbol[\reals][k]\defeq }
    \tPair{\var}{\fun{\var[2]}{\tTriple{\var'_1*\var[2]}{\ldots}{\var'_k*\var[2]}}}:     \Dsyn[k]{\reals} \to \Dsynrev[k]{\reals}\\
\evRsymbol [k]_\reals\defeq \lambda {\var[3]}.\,\pMatch{\var[3]}{\var}{\var'}{\tPair{\var}{\var'\, \underline{1}}}: \Dsynrev[k]{\reals}\To\Dsyn[k]{\reals} .
\end{array}$\\
Here, $\Dsynsymbol[k]$ is a macro on types (and terms) with
exactly the same inductive definition as $\Dsynsymbol$ except for the base case $\Dsyn[k]{\reals}=\bProd{\reals}{\reals^k}$.
By noting that both $\Dsynsymbol[k]$ and $\Dsynrevsymbol[k]$ preserve all type formers,
we can extend these definitions to all first order types $\ty$: $
    \var[3]:\Dsyn[k]{\ty}\vdash\lamR[\ty][k]{\var[3]}: \Dsynrev[k]{\ty},$
    $\var[3]:\Dsynrev[k]{\ty}\vdash\evR[\ty][k]{\var[3]}: \Dsyn[k]{\ty}$.
We can think of $\lamR[\ty][k]{\var[3]}$ as encoding $k$ tangent vectors $\var[3]:\Dsyn[k]{\ty}$ as a closure, so it is suitable
for running $\Dsynrev[k]{\trm}$ on, and $\evR[\ty][k]{\var[3]}$ as actually evaluating the reverse pass defined by 
$\var[3]:\Dsynrev[k]{\ty}$ and returning the result as  $k$ tangent vectors.
The idea is that given some $\var:\ty\vdash\trm:\ty[2]$ between first order types $\ty,\ty[2]$,
we run our continuation-based AD by running $\evR[{\ty[2]}][k]{\subst{\Dsynrev[k]{\trm}}{\sfor{\var}{\lamR[\ty][k]{\var[3]}}}}$.

The correctness proof closely follows that for forward AD.
In particular, one defines a binary logical relation $\semgl{\reals}^{r,k}=(\RR,\RR\times (\RR^k)^\RR, S^{r,k}_{\reals})$, where
$S^{r,k}_{\reals}=
\left\{ \left(f,x\mapsto \sPair{f(x)}{y\mapsto \sTriple{\partial_1 f(x) * y}{\ldots}{\partial_k f(x) * y}}\right)\mid f\in\plots{\RR}^{\RR^k}\right\}
$, on the plots $\plots{\RR}^{\RR^k}\times\plots{\RR\times ((\RR^k)^\RR)}^{\RR^k}$
and verifies that $\sem{\cnst}\times \sem{\Dsynrev[k]{\cnst}}$, $\sem{\var+\var[2]}\times \sem{\Dsynrev[k]{\var +\var[2]}}$, $\sem{\var*\var[2]}\times \sem{\Dsynrev[k]{\var *\var[2]}}$ and $\sem{\sigmoid(\var)}\times \sem{\Dsynrev[k]{\sigmoid(\var)}}$
respect this logical relation.
It follows that this relation extends to a functor $\semgl{-}^{r,k}:\Syn\to\Gl[\RR^k]$ such that $\id\times \Dsynrevsymbol[k]$
factors over $\semgl{-}^{r,k}$, implying the correctness of the continuation-based AD by the following lemma.
\begin{lemma}
    For all first order types $\ty$ (i.e. types not involving function types),
    we have that $\sem{\evR[\ty][k]{\lamR[\ty][k]{\trm} }}=\sem{\trm}$.
\end{lemma}
\begin{proof}[notes] This follows by an induction on the structure of $\ty$. The idea is that $\lamRsymbol[k]_{\ty}$ embeds reals
    into function spaces as linear maps, which is undone by $\evRsymbol[k]_{\ty}$ by evaluating the linear maps at $\underline{1}$. \qed
\end{proof}

To phrase correctness, in this setting, however, we need a few definitions.
Keeping in mind the canonical projection $\Dsem{M}\to M$, we define $\Dsem[k]{M}$ as the $k$-fold categorical pullback (fibre product) $\Dsem{M}\times_M\ldots \times_M \Dsem{M}$.
To be explicit, $\Dsemsymbol[k]_x{M}$ consists of $k$-tuples of tangent vectors at the base point $x$.
Again, $\Dsemsymbol[k]$ extends to a functor $\Man\to \Man$ by defining $\Dsem[k]{f}(x,\sTriple{v_1}{\ldots}{v_k})\defeq \sPair{f(x)}{\sTriple{\Dsemsymbol_x(f)(v_1)}{\ldots}{\Dsemsymbol_x(f)(v_k)}}$.
As $\Dsemsymbol[k]$ preserves countable coproducts and finite products (like $\Dsemsymbol$), it follows that the isomorphisms $\DtoT{\ty}$ generalize to 
canonical isomorphisms $\DtoT[,k]{\ty}:\sem{\Dsyn[k]{\ty}}\to \Dsem[k]{\sem{\ty}}$ for first order types $\ty$.
This leads to the following correctness statement for continuation-based AD.
\begin{theorem}[Semantic correctness of ${\Dsynrevsymbol[k]}$]
    \label{thm:rev-cor-full}
    For any ground $\ty$, any first order context $\Gamma$
    and any term $\Gamma\vdash\trm:\ty$,
    syntactic translation $\trm\mapsto \evR[{\ty}][k]{\subst{\Dsynrev[k]{\trm}}{\sfor{\ldots}{\lamR[\Gamma][k]{\var[3]}}}}$ coincides with the tangent bundle functor, modulo these canonical isomorphisms:
    \[\xymatrix{
        \sem{\Dsyn[k]\Gamma} \ar[rrr]^{\sem{\lamRsymbol[k]_{\Gamma};\Dsynrev[k]{\trm};\evRsymbol[k]_{\ty}}}  \ar[d]_{\DtoT[,k]{\Gamma}}^\cong
        &&&\sem{\Dsyn[k] \ty}\ar[d]^{\DtoT[,k]{\ty}}_\cong
        \\
        \Dsem[k]{\sem\Gamma}\ar[rrr]_{\Dsem[k]{\sem \trm}}
        &&&\Dsem[k]{\sem \ty}
        }\]
\end{theorem}

For example, when $\tau=\reals$ and $\Gamma= \var,\var[2]:\reals$, we can run our continuation-based
AD to compute the gradient of a program $\var,\var[2]:\reals\vdash \trm:\reals$ at values $\var=V,\var[2]=W$ by evaluating
\[
\subst{\evRsymbol[2]_{\reals}\,(\subst{\Dsynrev[2]{\trm}}{\sfor{\var}{(\lamRsymbol[2]_{\var:\reals}\,v)},\sfor{\var[2]}{(\lamRsymbol[2]_{\var[2]:\reals}\,w)}})}
{\sfor{v}{\tPair{V}{\tPair{\underline{1}}{\underline{0}}}},\sfor{w}{\tPair{W}{\tPair{\underline{0}}{\underline{1}}}}}.\]
Indeed,
\begin{align*}
    &\sem{\subst{\evRsymbol[2]_{\reals}\,(\subst{\Dsynrev[2]{\trm}}{\sfor{\var}{(\lamRsymbol[2]_{\var:\reals}\,v)},\sfor{\var[2]}{(\lamRsymbol[2]_{\var[2]:\reals}\,w)}})}
    {\sfor{v}{\tPair{V}{\tPair{\underline{1}}{\underline{0}}}},\sfor{w}{\tPair{W}{\tPair{\underline{0}}{\underline{1}}}}}}=\\
    &\big(\sem{t}(\sem{V},\sem{W}),\partial_1\sem{t}(\sem{V},\sem{W}),\partial_2\sem{t}(\sem{V},\sem{W})\big).
\end{align*} 


\section{Discussion and future work}\label{sec:conclusion}

\subsubsection{Summary.}
We have shown that diffeological spaces provide a denotational semantics for a higher order language with variants and inductive types (\S\ref{sec:semantics},\ref{sec:extended-language}). We have used this to show correctness of a simple AD translation (Thm.~\ref{thm:fwd-cor-basic}, Thm.~\ref{thm:fwd-cor-full}). But the method is not tied to this specific translation, as we illustrated in Section~\ref{sec:rev-mode-short}.

The structure of our elementary correctness argument for Theorem~\ref{thm:fwd-cor-basic} is a typical logical relations proof.
As explained in Section~\ref{sec:correctness}, this can equivalently be understood as a denotational semantics in a new kind of space obtained by categorical gluing.

Overall, then, there are two logical relations at play. One is in diffeological spaces, which ensures that all definable functions are smooth. The other is in the correctness proof (equivalently in the categorical gluing), which explicitly tracks the derivative of each function, and tracks the syntactic AD even at higher types.

\subsubsection{Connection to the state of the art in AD implementation.}
As is common in denotational semantics research, we have here focused on an idealized language and simple translations to illustrate the main aspects of the method. There are a number of points where our approach is simplistic compared to the advanced current practice, as we now explain.
\paragraph{Representation of vectors.}
\label{sub:discussion}
In our examples we have treated $n$-vectors as tuples of length~$n$. This style of programming does not scale to large~$n$. A better solution would be to use array types, following~\cite{shaikhha2019efficient}.
Our categorical semantics and correctness proofs straightforwardly extend to cover them, in a similar way to our treatment of lists.

\paragraph{Efficient forward-mode AD.}
For AD to be useful, it must be fast.
The syntactic translation $\Dsynsymbol$ that we use is the basis of an efficient AD library~\cite{shaikhha2019efficient}. However, numerous optimizations are needed, ranging from algebraic manipulations, to partial evaluations, to the use of an optimizing C compiler. A topic for future work would be to validate some of these manipulations using our semantics. The resulting implementation is performant in experiments~\cite{shaikhha2019efficient}. 

\paragraph{Efficient reverse-mode AD.}
Our sketch of continuation-based AD is primarily intended to emphasise that our denotational approach is not tied to any specific translation~$\Dsynsymbol$.
Nonetheless, it is worth noting that this algorithm shares similarities with advanced reverse-mode implementations: (1) it calculates 
derivatives in a (contravariant) ``reverse pass'' in which derivatives of operations are 
evaluated in the reverse order compared to their order in calculating the function 
value; (2) it can be used to calculate the full gradient of a function $\RR^n\to\RR$
in a single reverse pass (while $n$ passes of fwd AD would be necessary).
However, it lacks important optimizations and the continuation scales with the size of the input $n$ where it should scale with the size of the output. This adds an important overhead, as pointed out in \cite{pearlmutter2008reverse}. Speed being the main attraction of reverse-mode AD, its implementations tend to rely on mutable state, control operators and/or staging \cite{pearlmutter2008reverse,carpenter2015stan,wang2018demystifying,brunel2019backpropagation}, 
which we have not considered here.

\paragraph{Other language features.}
\label{sub:summary_and_future_work}
The idealized languages that we considered so far do not touch on several useful language constructs. For example: the use of functions that are partial (such as division) or partly-smooth (such as RelU); phenomena such as iteration, recursion; and probabilities.  
There are suggestions that the denotational approach using diffeological spaces can be adapted to these features using standard categorical methods. We leave this for future work.

\subsubsection{Acknowledgements.}
We have benefited from discussing this work with many people, including B.~Pearlmutter, O.~Kammar, C.~Mak, L.~Ong, G.~Plotkin, A.~Shaikhha, J.~Sigal, and others. 
Our work is supported by the Royal Society and by a Facebook Research Award. In the course of this work, MV has also been employed at Oxford (EPSRC Project EP/M023974/1) and at Columbia in the Stan development team.
This project has received funding from the European Union’s Horizon 2020 research and innovation
programme under the Marie Skłodowska-Curie grant agreement No. 895827.

\clearpage


%
%
%
\bibliographystyle{splncs04}
\bibliography{bibliography}


\vfill

{\small\medskip\noindent{\bf Open Access} This chapter is licensed under the terms of the Creative Commons\break Attribution 4.0 International License (\url{http://creativecommons.org/licenses/by/4.0/}), which permits use, sharing, adaptation, distribution and reproduction in any medium or format, as long as you give appropriate credit to the original author(s) and the source, provide a link to the Creative Commons license and indicate if changes were made.}

{\small \spaceskip .28em plus .1em minus .1em The images or other third party material in this chapter are included in the chapter's Creative Commons license, unless indicated otherwise in a credit line to the material.~If material is not included in the chapter's Creative Commons license and your intended\break use is not permitted by statutory regulation or exceeds the permitted use, you will need to obtain permission directly from the copyright holder.}

\medskip\noindent\includegraphics{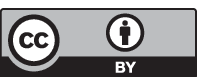}

\clearpage
\ifx\fossacsversion\undefined\appendix\section{$\CartSp$ and $\Man$ are not cartesian closed categories}
\label{sec:man_not_ccc}
\begin{lemma}\label{lem:borsuk}
There is no continuous injection $\RR^{d+1}\to \RR^d$.
\end{lemma}
\begin{proof}
If there were, it would restrict to a continuous injection
$S^d\to \RR^d$.
The Borsuk-Ulam theorem, however, tells us that every continuous 
$f:S^d\to\RR^d$ has some $x\in S^d$ such that $f(x)=f(-x)$, which is a
contradiction. \qed
\end{proof}
Let us define the terms:
\[
\var_0:\reals,\ldots, \var_n:\reals\vdash \trm_n=
\fun{\var[2]}{\var_0+\var_1 * y + \dots + \var_n*y^n}:\reals\To\reals
\]
Assuming that $\CartSp$/$\Man$ is cartesian closed, observe that these get
interpreted as injective continuous (because smooth)
functions $\RR^n\to \sem{\reals\To\reals}$ in $\CartSp$ and $\Man$.
\begin{theorem}
$\CartSp$ is not cartesian closed.
\end{theorem}
\begin{proof}
In case $\CartSp$ were cartesian closed, we would have $\sem{\reals\To\reals}=\reals^n$
for some $n$.
Then, we would get, in particular a continuous 
injection $\sem{\trm_{n+1}}:\RR^{n+1}\to \RR^n$,
which contradicts Lemma \ref{lem:borsuk}. \qed
\end{proof}
\begin{theorem}
$\Man$ is not cartesian closed.
\end{theorem}
\begin{proof}
Observe that we have $\iota_n:\RR^n\to \RR^{n+1}$; $\tTriple{a_0}{\ldots}{a_n}\mapsto 
\tTuple{a_0,\ldots,a_n,0}$
and that $\iota_n;\sem{\trm_{n+1}}=\sem{\trm_n}$.
Let us write $A_n$ for the image of $\sem{\trm_n}$ and $A=\cup_{n\in\NN}A_n$.
Then, $A_n$ is connected because it is the continuous image of a connected set.
Similarly, $A$ is connected because it is the non-disjoint union of connected sets.
This means that $A$ lies in a single connected component of $\sem{\reals\To\reals}$,
which is a manifold with some finite dimension, say $d$.

Take some $x\in\RR^{d+1}$ (say, $0$), take some open $d$-ball $U$ around $\sem{\trm_{d+1}}(x)$,
and take some open $d+1$-ball $V$ around $x$ in $\sem{\trm_{d+1}}^{-1}(U)$.
Then, $\sem{\trm_{d+1}}$ restricts to a continuous injection from $V$ to $U$, or equivalently,
$\RR^{d+1}$ to $\RR^d$, which contradicts Lemma \ref{lem:borsuk}. \qed
\end{proof}\fi
\end{document}